\documentclass[aoas,preprint]{imsart}

\RequirePackage[OT1]{fontenc}
\RequirePackage{amsthm,amsmath}
\RequirePackage{natbib,enumerate}
\RequirePackage[colorlinks,citecolor=blue,urlcolor=blue]{hyperref}
\usepackage{bbold}
\usepackage[normalem]{ulem}
\usepackage{tikz}
\usepackage{pgfplots}
\usepackage[ruled,vlined]{algorithm2e}

\usepackage{subcaption}
\usepackage{float}
\usepackage{subcaption}
\usepackage{caption}

\usepackage[format=plain,labelfont=bf,textfont=it, font=small]{caption}

\usepackage{graphicx}
\usepackage{amsfonts}
\usepackage{wrapfig}
\usepackage{comment}
\usepackage{tabularx}
\newcommand{\Prob}{\ensuremath{{\mathbb P}}} 

\newcommand{\R}{\ensuremath{{\mathbb R}}}
\newcommand{\E}{\ensuremath{{\mathbb E}}}

\usepackage{lscape}

\newcommand{\Ind}{\ensuremath{{\mathds{1}}}} 
\usepackage{dsfont}
\usepackage{caption}
\usepackage{float}

\newcommand{\distas}[1]{\mathbin{\overset{#1}{\kern\z\sim}}}%
\newsavebox{\mybox}\newsavebox{\mysim}
\newcommand{\distras}[1]{%
  \savebox{\mybox}{\hbox{\kern3pt$\scriptstyle#1$\kern3pt}}%
  \savebox{\mysim}{\hbox{$\sim$}}%
  \mathbin{\overset{#1}{\kern\z@\resizebox{\wd\mybox}{\ht\mysim}{$\sim$}}}%
}

\usepackage{varioref}

\arxiv{arXiv:2106.03742}

\startlocaldefs
\numberwithin{equation}{section}
\theoremstyle{plain}
\newtheorem{proposition}{Proposition}[section]
\newtheorem{lemma}{Lemma}[section]
\newtheorem{definition}{Definition}[section]
\endlocaldefs

\begin{document}

\begin{frontmatter}
\title{Imputation Scores} 
\runtitle{Imputation Scores}
\thankstext{T1}{Authors with equal contribution.}

\begin{aug}
\author{\fnms{Jeffrey} \snm{Näf}*}, 
\author{\fnms{Meta-Lina} \snm{Spohn}*},
\author{\fnms{Loris} \snm{Michel}}
\and
\author{\fnms{Nicolai} \snm{Meinshausen}}



\address{Seminar for Statistics\\
ETH Zürich\\
Rämistrasse 101\\
8092 Zürich\\
Switzerland\\
E-mail: metalina.spohn@stat.math.ethz.ch}

\end{aug}


\begin{keyword}
\kwd{Ranking}
\kwd{Random Projections}
\kwd{Tree Ensembles}
\kwd{Random Forest}
\kwd{KL-Divergence}
\end{keyword}


\begin{abstract}
Given the prevalence of missing data in modern statistical research, a broad range of methods is available for any given imputation task. How does one choose the `best' imputation method in a given application? The standard approach is to select some observations, set their status to missing, and compare prediction accuracy of the methods under consideration of these observations. Besides having to somewhat artificially mask observations, a shortcoming of this approach is that imputations based on the conditional mean will rank highest if predictive accuracy is measured with quadratic loss. In contrast, we want to rank highest an imputation that can sample from the true conditional distributions. In this paper, we develop a framework called ``Imputation Scores'' (I-Scores) for assessing missing value imputations. We provide a specific I-Score based on density ratios and projections, that is applicable to discrete and continuous data. It does not require to mask additional observations for evaluations and is also applicable if there are no complete observations. The population version is shown to be proper in the sense that the highest rank is assigned to an imputation method that samples from the correct conditional distribution. The propriety is shown under the \emph{missing completely at random} (MCAR) assumption but is also shown to be valid under \emph{missing at random} (MAR) with slightly more restrictive assumptions. We show empirically on a range of data sets and imputation methods that our score  consistently ranks true data high(est) and is able to avoid pitfalls usually associated with performance measures such as RMSE. Finally, we provide the \textsf{R}-package \texttt{Iscores} available on CRAN with an implementation of our method.
\end{abstract}

\end{frontmatter}

\section{Introduction}

Missing data is a widespread problem in both research and practice. In the words of \citet{Meng1994}, ``imputing incomplete observations is becoming an indispensable intermediate phase between the two traditional phases [...] of collecting data and analyzing data.'' With the increasing dimensionality and size of modern data sets, the problem of missing data only gets more pronounced. 
One reason for this is that even for moderate dimensionality and overall fraction of missing entries, the set of complete observations tends to be small if not zero. As such, one would often have to discard a substantial part of the data when keeping only complete observations. In addition, depending on the missingness mechanism, working with complete cases is invalid in many situations.

Consequently, there is a large body of literature on imputation methods, filling in the missing entries in a given data set, see e.g., \citet{RubinLittlebook}. Such methods include the very general \emph{Multivariate Imputation by Chained
Equations} (mice) approach \citep{mice, Deng2016}, \emph{missForest} \citep{stekhoven_missoforest}, and \emph{multiple imputation in principal component analysis} (mipca) \citep{josse_mipca}. In more recent work, methods based on non-parametric Bayesian strategies \citep{Bayesianmethods}, generative adversarial networks \citep{Gain} and optimal transport \citep{josse} were developed. These are just a few examples. The `R-miss-tastic' platform, developed in an effort to collect knowledge and methods to streamline the task of handling missing data, lists over 150 packages \citep{mayer2021rmisstastic}.\

Despite the broad range of imputation methods, not a lot of attention has been paid in the literature to the question of how to evaluate and choose an imputation method for a given data set. If the true data underlying the missing entries are available (achievable if observations are artificially masked), the imputed values are often simply compared  to the true ones through the root mean-squared error (RMSE) or mean absolute error (MAE) (see e.g., \citet{stekhoven_missoforest}, \citet{Waljee2013}, \citet{Gain}, \citet{josse} and many others). For categorical variables, the percentage of correct predictions (PCP) can be used \citep{linreview}. To select an imputation method, the one with the lowest overall error value is chosen. This simple method is common, despite having major drawbacks. For example, RMSE (respectively, MAE and PCP) favors methods that impute with the conditional means (respectively, conditional medians and conditional modes), versus samples from the true conditional distributions \citep{gneiting_point}. As outlined in \citet[Chapter 2.5.1]{VANBUUREN2018}, this may lead to a choice of ``nonsensical'' imputation methods. In particular, they tend to artificially strengthen the association between variables, which can lead to invalid inference. A motivating example of such a case is given in Section \ref{motivatingExample}.

Another approach is to fix a \emph{target quantity}, calculated once on the full data and once on the imputed data and use suitable distances between them for a quality assessment of the imputation. The target quantity can be an expectation, a regression or correlation coefficient, a variance \citep{horton, mice-rf-method, josse_baysmipca, midastouch, xu} or more involved estimands such as treatment effects using propensity scores \citep{Choi}. While this approach is sensible if (one) target quantity can clearly be defined, this may not be ideal in many applications. Often it is not even clear beforehand what the target of interest is, or one deals with several targets of interests, each of which might lead to a different choice of imputation methods \citep[Chapter 2.3.4]{VANBUUREN2018}. 

As such, there are situations where one might prefer to have a tool that is able to identify a good imputation method for a wide range of targets. A way to achieve this is to simply define the notion of target measure more broadly. As noted in \cite{josse}: ``A desirable property of imputation methods is that they should preserve the joint and marginal distributions.'' That is, if $x^*$ is a complete observation and $z$ an observation with imputed values, we want them to be realizations of the same distribution. This can be seen as a special case of the above target quantity approach, with the target being the underlying joint distribution from which the data were generated. 

With this goal in mind, more elaborate distributional metrics, such as $\phi$-divergences \citep{phidivergence} or integral probability metrics (see e.g., \citet{Sriperumbudur2012}) are necessary, even when simply comparing an imputed data set to the true one. This was utilized in \cite{josse} with the Wasserstein Distance (WD), but seems otherwise uncommon, despite the drawbacks of measures like RMSE and MAE.

In applications, the true data for the missing values are rarely available. Ad-hoc methods have been proposed to assess the success of the imputation; one such approach is to mask some observed values, impute them, and then compare (via, say, RMSE) the imputations with the observed values. In this paper, we aim to assess more directly the quality of an imputation method. We propose the flexible framework of proper Imputation Scores (I-Scores) to evaluate imputation methods when (i) the target measure is the true data distribution, (ii) the true data underlying the missing values are not available, and (iii) we do not want to artificially mask observations for the evaluation. To overcome the difficulty of observing only incomplete data, we use random projections in the variable space. We propose a specific estimator of an I-Score based on density ratios (DR I-Score). The density ratio is estimated through classification, where the classifier acts as a discriminator between observed and imputed distribution. The result is an easy-to-use imputation scoring method. Under suitable assumptions on the missingness mechanism, we prove propriety of our DR I-Score, meaning that the true underlying data distribution is ranked highest in expectation. The propriety is shown under a \emph{missing completely at random} (MCAR) assumption on the missingness mechanism. It is also shown to be valid under \emph{missing at random} (MAR) if we restrict the random projections in the variable space to always include all variables, which in turn requires access to some complete observations. Empirical results show that indeed true data samples are generally ranked highest by the proposed I-Score if compared with widely used imputation methods. 
We provide an implementation of the method in the \textsf{R}-package \texttt{Iscores}, available on CRAN and GitHub (\url{https://github.com/missValTeam/Iscores}). The Supplementary Material \citep{oursupplement} contains the code used to generate the empirical analysis.

Section~\ref{problemformulationsection} introduces the notation and Section~\ref{motivatingExample} showcases a motivating example. Section~\ref{scoredefsec} defines the framework
of Imputation Scores (I-Scores) along with a specific I-Score, presenting a way of evaluating imputation methods in the presence of missing values. Section~\ref{pracaspects} then details the estimation of the I-Score and describes the algorithm. Section~\ref{relwork} presents further related work. Section~\ref{empresults_1} empirically validates the estimated I-Score on a range of data sets and on a real data set with missing values, and Section~\ref{sec:discussion} concludes with a discussion.

\section{Notation}\label{problemformulationsection}


We assume an underlying probability space $(\Omega, \mathcal{A}, \Prob)$ on which all random elements are defined. Throughout, we take $\mathcal{P}$ to be a collection of probability measures on $\R^d$, dominated by some $\sigma$-finite measure $\mu$. We denote the (unobserved) complete data distribution by $P^* \in \mathcal{P}$ and by $P$ the actually observed distribution with missing values. We assume that $P$ ($P^*$) has a density $p$ ($p^*$). We take $X$ ($X^*$) to be the random vector with distribution $P$ ($P^*$) and let $x_{i}$ ($x^*_{i}$), $i=1,\ldots, n$, be realizations of an i.i.d. copy of the random vector $X$ ($X^*$). Similarly, $M$ is the random vector in $\{0,1\}^d$, encoding the missingness pattern of $X$, with realization $m$, whereby for  $j=1,\ldots,d$, $m_j=0$ means that variable $j$ is observed, while $m_j=1$ means it is missing. For instance, the observation $(\texttt{NA}, x_2,x_3)$ corresponds to the pattern $(1,0,0)$. We denote the distribution of $M$ as $P^M$, with support $\mathcal{M}$, so that $\Prob(M=m)=P^M(m)$. 

For a subset $A \subseteq \{1,\ldots,d\}$ and for a random vector $X$ or an observation $x$ in $\R^d$, we denote with $X_A$ ($x_A$) its projection onto that subset of indices. 
For instance if $d=3$ and $A=\{1,2\}$, then $X_A=(X_1, X_2)$ ($x_A=(x_1, x_2)$). The projection onto $A$ of the observation $x_i$, $(x_i)_A$, is denoted as $x_{i,A}$. Analogously, for a missingness pattern $M \sim P^M$ or an observation $m$ in $\{0,1\}^d$, we denote with $M_A$ $(m_A)$ its projection onto the subset of indices in $A$. If $X$ has a density $p$ on $\R^d$, we denote by $p_{A}$ the density of the projection $X_A$. 

To denote assumptions on the missingness mechanism, we use a notation along the lines of \cite{whatismeant}. For each realization $m$ of the missingness random vector $M$ we define with $o(X,m):=(X_j)_{j \in \{1,\ldots,d\}:m_j=0}$ the observed part of $X$ according to $m$ and with $o^c(X,m):=(X_j)_{j \in \{1,\ldots,d\}:m_j=1}$ the corresponding missing part. Note that this operation only filters the corresponding elements of $X$ according to $m$, regardless whether or not these elements are actually missing or not. For instance, we might consider the unobserved part $o^c(X,m)$ according to $m$ for the fully observed $X$, that is $X \sim P|M=\mathbf{0}$, where $\mathbf{0}$ denotes the vector of zeros of length $d$.

\section{Motivating Example} \label{motivatingExample}

\begin{figure}
    \centering
    \begin{subfigure}[b]{1\textwidth}
    \centering
        \includegraphics[width=1\linewidth]{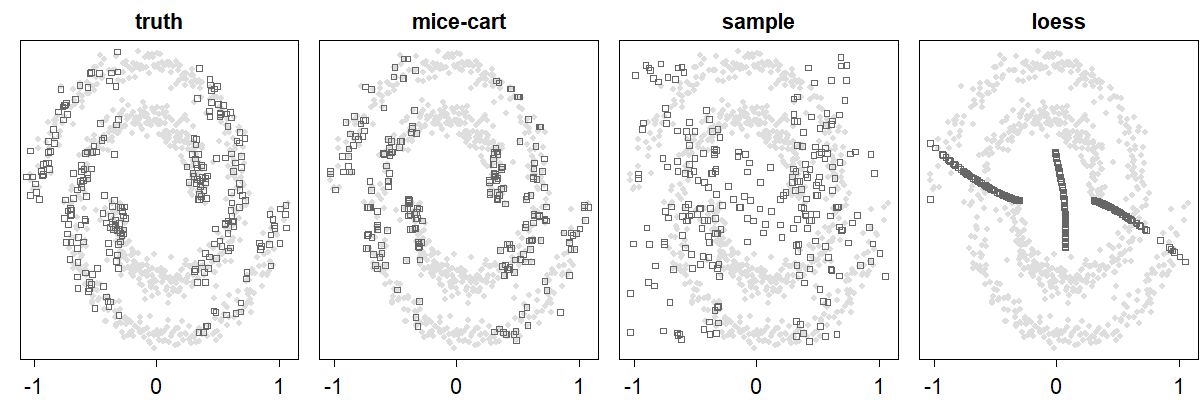}
        \caption{MAR}
    \end{subfigure}
    \begin{subfigure}[b]{1\textwidth}
    \centering
       \includegraphics[width=1\linewidth]{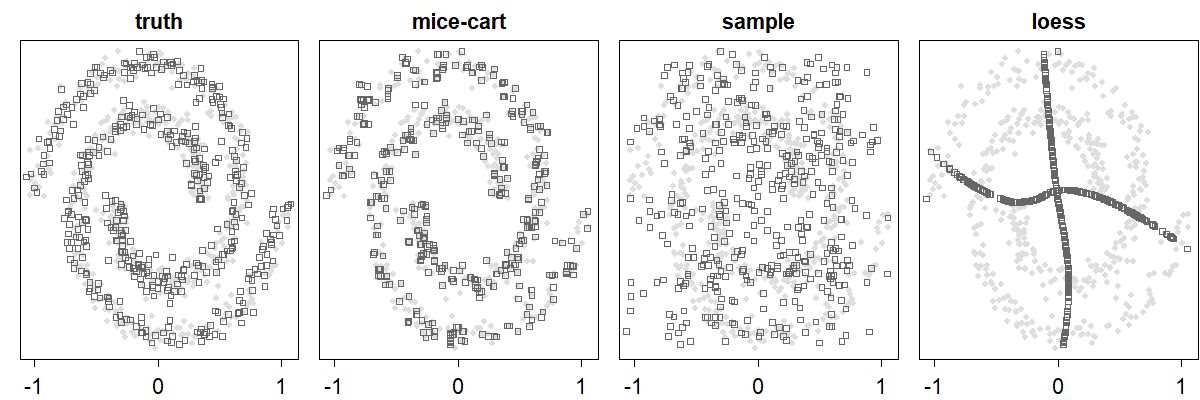}
        \caption{MCAR}
    \end{subfigure} 
    \caption{Imputations for the spiral example $(n=1000,$ $p=2)$. The complete observations are plotted in light gray with dots and the imputed observations in dark gray with squares. From left to right, true data, mice-cart, sample and loess are shown. The top row shows the MAR case, the bottom row the MCAR case. }
    \label{fig:my_label}
\end{figure}
As a toy motivating example we consider a noisy version of the spiral in two dimensions from the \textsf{R}-package \texttt{mlbench} \citep{mlbench}. The concepts and tools introduced here will be detailed and used in the remainder of the paper. We generated $1000$ observations of the noisy spiral, each entry having a probability of being missing of $p_{miss}=0.3$ with missingness mechanisms MAR and MCAR. In the case of MAR, the variable $X_2$ is missing with probability $p_{miss}$ if the corresponding $X_1$ is $> 0.3$ or $ < -0.3$ and observed otherwise. The variable $X_1$ is missing with probability $p_{miss}$, if $X_2 \in [-0.3,0.3]$. In the case of MCAR, we set every value in the data matrix to \texttt{NA} with probability $p_{miss}$. For MCAR we face already in this low dimensional example an average of around $(1-(1-0.3)^2)*1000 \approx 500$ observations with at least one missing value, which is half the sample size. 

We might decide to impute the missing values and do so with three methods: i) simply estimating the conditional expectation $\E[X_1|X_2]$ and $\E[X_2| X_1]$ on the complete cases using a local regression approach (``loess''), and filling in the missing values by predicting from $X_1$ (if $X_2$ is missing) or from $X_2$ (if $X_1$ is missing), ii) random sampling an observed value for each missing entry (``sample''), and iii) mice \citep{mice} combined with a single tree in each iteration (``mice-cart''). See Appendix \ref{impmethodsdetail} for more explanation on sample and mice-cart.

Though a very artificial example, it highlights some interesting features of different evaluation methods of the imputations. As mentioned in the introduction, our target is $P^*$, the full distribution of the data. In this two dimensional example, a visual evaluation is possible. Figure~\ref{fig:my_label} shows the resulting imputations. While mice-cart and the true underlying data are hard to distinguish, it is apparent that sample is worse than either and loess in turn is much worse than sample.

We may now try to quantify this visually obtained ordering. For the three imputations as well as the true underlying data we compute our DR I-Score (defined later). This score is positively oriented: A higher value indicates a better performance. We additionally compute the negative of RMSE (``negRMSE''), where the negative sign assures the same orientation of higher values indicating better performance. We emphasize that negRMSE is computed using the unobserved full data set, as commonly done in research papers introducing new methods of imputation. We also computed approximated two-sided 95\%-confidence intervals (CI) of the DR I-Score and negRMSE, as detailed in Section \ref{pracaspects}. The results are shown in Figure~\ref{fig:geo}. In the left plot (a), the score of the true data was subtracted from the scores to visualize the comparison to the true data imputation. In the right plot (b), normalizing negRMSE scores, using the true data imputation, is unnecessary, since negRMSE is $0$ by definition for the true data imputation.


Maybe unsurprisingly, negRMSE appears to poorly measure the distributional difference between imputed and real data set. In particular, for MCAR and MAR its value is highest for the loess imputation, and significantly so, based on the approximated CI. This is despite the fact that negRMSE is allowed to use the unobserved data. In contrast, our proposed DR I-Score has no access to the unobserved data, and nonetheless ranks true data clearly highest, with mice-cart as a close second, followed by sample and loess. Thus, without using the unobserved samples, our score manages to give a sensible ranking in this example that is in line with the visual impression, for both a MCAR and MAR missingness mechanisms. 



\begin{figure}
     \centering
     \begin{subfigure}[b]{0.46\textwidth}
         \centering
         \includegraphics[width=1\textwidth]{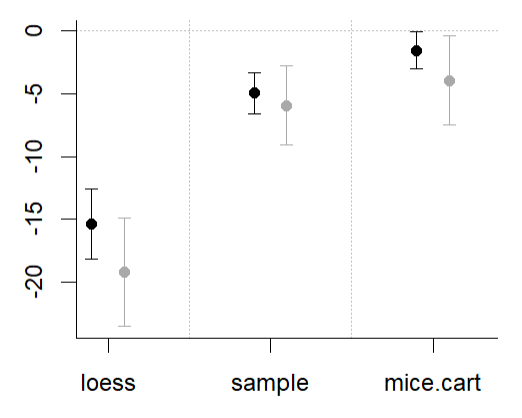}
              \caption{Estimated DR I-Score with CIs.}
     \end{subfigure}
     \begin{subfigure}[b]{0.46\textwidth}
         \centering
         \includegraphics[width=1\textwidth]{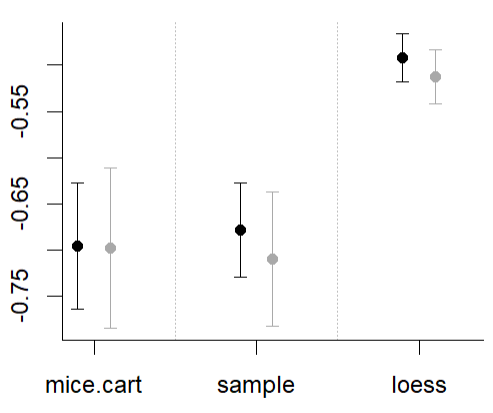}
         \caption{Negative RMSE with CIs.}
     \end{subfigure}
     \caption{Spiral example ($n=1000$, $p=2$):  Estimation of the proposed DR I-Score (a) and the negRMSE (b) with corresponding CIs for the methods mice-cart, sample and loess under the missingness mechanisms MCAR (black) and MAR (grey). We obtained the CIs by subsampling as described in Section \ref{section:jackknife}. In (a) we subtracted the score of the true data from the scores of the methods, thus the line at $0$ represents the true data score. We used $p_{miss}=0.3$ to generate missing values. }
\label{fig:geo}

     \end{figure}

\section{Scores for Imputations} \label{scoredefsec}


Our notion of proper Imputation Score (I-Score) is inspired by the classical notion of proper scores (or proper scoring rules) and proper divergence functions (see e.g. \citet{gneiting, gneiting_point,Scoresforclimate}). A score assesses whether a probabilistic forecast is close to the true distribution. In a traditional sense, a score $S$ takes a predictive distribution $Q_1 \in \mathcal{P}$ as a first argument, as well as $X \sim Q_2 \in \mathcal{P}$ as a second: $S(Q_1,X)$. The corresponding expectation over the distribution $Q_2$ is denoted by $S(Q_1,Q_2 ):= \E_{X \sim Q_2}[S(Q_1,X)]$. Thus, a score value is assigned to $Q_1$ over the comparison with $Q_2$. A desirable natural property is that $ S(Q_1,Q_2)\leq S(Q_2,Q_2) $, for all $Q_1,Q_2 \in \mathcal{P}$. This ensures that the true distribution $Q_2$ of $X$ is scored highest in expectation. A score that meets this requirement is referred to as proper.

\subsection{Imputation Score (I-Score)} 


We now define the notion of a proper I-Score. Despite the analogy to the classical notion of scores, we will need some deviations for the setting of imputation scores as applied to partially observed samples. Recall that $P$ refers to the distribution of $X$ with missing values and correspondingly, $P^* \in \mathcal{P}$ refers to the distribution of $X^*$ without missing values. We denote with $P^M$ the distribution of $M$ with domain $\mathcal{M}$. Naturally, $(P^*, P)$ form a tuple, whereby $P$ is derived from $P^*$ and $P^M$. We define $\mathcal{H}_{P} \subset \mathcal{P}$ to be the set of imputation distributions compatible with $P$, that is 
\begin{align}\label{imputationdistset}
      \mathcal{H}_{P}:= \{H \in \mathcal{P}: h(o(x,m)|M=m)=p(o(x,m)|M=m) \text{ for all } m \in \mathcal{M} \},
\end{align}
where as above for a pattern $m$, $o(x,m)=(x_j)_{j \in \{1,\ldots,d\}:m_j=0}$ subsets the observed elements of $x$ according to $m$, while $o^c(x,m)=(x_j)_{j \in \{1,\ldots,d\}:m_j=1, }$ subsets the missing elements\footnote{Note that while $h$ and $p$ are densities on $\R^d$, notation is slightly abused by using expressions such as $h(o(x,m)|M=m)$ and $p(o(x,m)|M=m)$, which are densities on $\R^{|\{j:m_j=0 \}|}$.}. Clearly, $P^* \in \mathcal{H}_{P}$, so that the true distribution $P^*$ can be seen as an imputation.



\begin{definition} \textbf{Imputation Score (I-Score)}  \label{Iscoredef}

Let $(P^*, P)$ be the tuple of distributions as described above and $H \in \mathcal{H}_P$, as given in \eqref{imputationdistset}. 
A real-valued function $S_{NA}(H, P)$
is a proper I-Score iff
$$ S_{NA}(H, P) \leq S_{NA}(P^*, P), 
$$
for any imputation distribution $H\in \mathcal{H}_{P}$. It is strictly proper iff the inequality is strict for $H \neq P^*$.
\end{definition}

In practice, we do not have access to $H$ and $P$ but can only draw samples from $(H,P)$ by drawing a sample $x$ from $P$ (with missingness pattern $m$) and the corresponding coupled imputation $z$ from $H$ that matches $x$ on all observed variables, so that $o(z,m) = o(x,m)$. In the ideal case where $(Z,X)$ are drawn from $(P^{*},P)$, the conditional distribution of $o^c(Z, M)$ given $o(Z,M)$ is identical to the conditional distribution of $o^c(X^*,M)$ given $o(X^*,M)$.
In principle, we could also score the full conditional distribution of imputation algorithms instead of just supposing that we can draw samples from the conditional distribution. However, most popular imputation algorithms do not specify the conditional distribution explicitly and we thus focus on the latter case. We discuss the sample-based implementation in Section \ref{pracaspects}.

A key difference to the classical notion of scores \citep{gneiting, gneiting_point} is that we do not observe a sample from the `true' distribution of the missing values of $x$, given missingness pattern $m$ and observed values $o(x,m)$. We can hence not compare the 
predictions of the missing values (whether they are explicit distributional predictions or samples from the imputation) to the realization of a sample from the underlying `true' conditional distribution. It seems perhaps surprising that we can still obtain proper I-Scores in this setting. A central assumption is that the missingness mechanism is such that any subset of variables has a positive probability of being observed (in contrast to a classical regression setting where the target variable is always unobserved before the prediction is made). The variability of the missingness patterns hence allows to construct proper I-Scores. An alternative approach would be to mask observed values as unobserved and score the imputation on these held-out data. This would require to change the distribution of the missingness patterns, however, which would be especially problematic in absence of a MCAR assumption.

We immediately see that negRMSE used in the motivating example of Section \ref{motivatingExample} does not fit into the framework of proper I-Scores, simply because it requires access to the true underlying values according to $P^*$.

We now define a specific I-Score that satisfies Definition \ref{Iscoredef}, the density ratio (DR) I-Score.






\subsection{Density Ratio I-Score}

The goal is to provide a methodological framework to construct an I-Score based on density ratios, the DR I-Score. In addition, we try to circumvent the problem of not observing $P^*$ in a data efficient way, employing \emph{random projections} in the variable space. Considering observations that are projected into a lower-dimensional space allows us to recover more complete cases of the true underlying distribution. As an example, $x = (\texttt{NA},1,\texttt{NA},2)$ is not complete, however if we project it to the dimensions $A = \{2,4\}$, $x_A=(1,2)$ is complete in this lower-dimensional space. In what follows, we average distributional differences between imputed and complete cases over different projections. In fact, these projections constitute an additional source of randomness over which we integrate to obtain the DR I-score, as we detail now.

Let $\mathcal{A}$ be a subset of the power set $2^{\{1,\ldots ,d \}}$, that denotes the set of all possible projections, such that each $A\in \mathcal{A}$ describes a set of variables we project onto. Assume the projections are chosen randomly according to a distribution $\mathcal{K}$ with support $\mathcal{A}$ (see Section \ref{distributionkappa} for more details on $\mathcal{K}$). Similarly, we define for each $A$ a distribution $P^M_A$ over the missingness patterns in $P_A$ with support $\mathcal{M}_A$. That is, $m_A \in \mathcal{M}_A$ is a given missingness pattern on the projection with associated probability $\Prob(M_A=m_A)=P^M_A(m_A)$. For any distribution $H \in \mathcal{H}_{P}$ we can then consider the conditional distribution $H_A| M_A=m_A$, i.e. the distribution of an imputation $H$, given the missingness pattern $m_A$ on the projection $A$. We will abbreviate this distribution with $H_{m_A}$, so that the density of $H_{m_A}$ is given as $h_{m_A}(x_A):= h_A(x_A|M_A=m_A)$. Denoting with $\mathbf{0}$ the vector of zeros, we similarly write $p_A(x_A|M_A=\mathbf{0})$ for the density of the fully observed points on the projection $A$.


The following definition specifies the density ratio I-Score as an expected value of the log-ratio of the two densities, where the expectation is also taken over $A$ and $M_A$.


\begin{definition}\label{phidef} \textbf{Density Ratio I-Score}\\
We define the DR I-Score of the imputation distribution $H$ by
    \begin{equation}
       S^{*}_{NA}(H, P)= \E_{A \sim \mathcal{K}, M_A \sim P_A^M, X_A \sim H_{M_A}}\left[\log \left(\frac{p_A(X_A \mid M_A=\mathbf{0})}{h_{M_A}(X_A)} \right)\right]. \label{drscore}
    \end{equation}
\end{definition} 

If $p_A(X_A \mid M_A=\mathbf{0})=0$ for a set of $X_A \sim H_{M_A}$ with nonzero probability, we take $S^{*}_{NA}(H, P)=-\infty$ as a convention. As an intuition, the DR I-Score given by \eqref{drscore} can be rewritten as
      \begin{equation*}
          S_{NA}^{*}(H,P) = - \E_{A \sim  \mathcal{K}, M_A \sim P_A^M}D_{KL}(h_{M_A} \mid \mid p_A(\cdot \mid M_A=\mathbf{0})), 
      \end{equation*}
     where the Kullback-Leibler divergence (KL divergence) between two distributions $P,Q \in \mathcal{P}$ on $\R^d$ with densities $p$, $q$ is defined by
     \begin{equation*}
         D_{KL}(p \mid \mid q) := \int p(x) \log \left( \frac{p(x)}{q(x)}\right)d \mu(x).
     \end{equation*}


To prove that the DR I-Score is indeed a proper I-Score we need an assumption on the missingness mechanism. This is shown in the following result:

\begin{proposition}\label{amazingproposition}
Let $H\in \mathcal{H}_{P}$, as defined in \eqref{imputationdistset}. If for all $A\in \mathcal{A}$,
    \begin{align}\label{MARcond}
       &p^*(o^c(x_A,m_A)|o(x_A,m_A), M=m'_A) = p^*(o^c(x_A,m_A)|o(x_A,m_A)), \nonumber \\
       &\text{ for all } m'_A, m_A \in \mathcal{M}_A,
    \end{align}
 then $S^{*}_{NA}(H, P)$ in \eqref{drscore} is a proper I-Score.
\end{proposition}

The proof is given in the Appendix.

Condition \eqref{MARcond} is simply the MAR condition on the projection $A$, see e.g., \citet{whatismeant}.\footnote{Lemma \ref{MARcond} in the Appendix shows that condition \eqref{MARcond} is indeed equivalent to the MAR condition in \citet{whatismeant}.} The key insight is that for any imputation distribution $H \in \mathcal{H}_P$ and $m \in \mathcal{M}$ it holds that
\begin{align}\label{hfactorization}
 h_m(x)&=h(o^c(x, m)| o(x, m), M=m) p^*(o(x,m)| M=m),
\end{align}
by the definition of $\mathcal{H}_P$ in \eqref{imputationdistset}. This can be used to show that the score in \eqref{drscore} factors into (i) an \emph{irreducible part}, stemming from the difference in the observed parts $p^*(o(x,m)| M=m)$ and $p(o(x,m)| M=\mathbf{0})$, and (ii) a score for the distance of the conditional distributions. The latter is minimized for $H=P^*$.


Thus Proposition \ref{amazingproposition} shows that the proposed I-Score is proper as long as the missingness mechanism is MAR on each projection $A \in \mathcal{A}$. In particular, this holds if
\begin{itemize}
    \item[(i)] the missingness mechanism is MCAR\footnote{Flexible nonparametric methods to test MCAR were developed recently, see e.g., \citet{Li2015, 2021pklm}.},
    \item[(ii)] the missingness mechanism is MAR and $\mathcal{A}=\{1,\ldots,p \}$,
    \item[(iii)] it is known that blocks of data are jointly MAR, and the set of projections $\mathcal{A}$ is chosen such that the blocks are contained as a whole in the projections.
\end{itemize}

As will be discussed in Section \ref{pracaspects}, we face a trade-off in practice: The method tends to have more power when allowing for smaller projections while this increases the chance of having a lot of projections violating \eqref{MARcond}, which may hurt the propriety of the score. Nonetheless, in the empirical validation (Section 7) we do not find evidence that our score violates propriety, even when using random projections without any verification of MAR in the projections.

While our DR I-Score uses a KL-based measurement of the difference between the true distribution and imputation distributions, it would be interesting to look into alternatives such as the multivariate generalization of the integrated quadratic
distance of \cite{Scoresforclimate}, based on the energy score of \cite{Gneiting2008}, as suggested by a referee.

\subsection{Assessing a Density Ratio through Classification} \label{sec:drclass}



We assess the density ratio of the proposed DR I-Score \eqref{drscore} by classification, as e.g. in \cite{breiman2003} or \cite{Cal2020}. Given a projection $A$ and a pattern $m_A$, we define for $x_A$ in the support of $H_{m_A}$,
 \begin{equation}
         \pi_{m_A}(x_A):= \frac{p_A(x_A \mid M_A=\mathbf{0})}{p_A(x_A \mid M_A=\mathbf{0}) + h_{m_A}(x_A)}.  \label{def_pA}
\end{equation}
It then follows that we can rewrite the density ratio in (\ref{drscore}) as
\begin{equation}
   \frac{p_A(x_A \mid M_A=\mathbf{0})}{h_{m_A}(x_A)} = \frac{\pi_{m_A}(x_A)}{1-\pi_{m_A}(x_A)}, \label{drformula} 
\end{equation}
leaving us with the problem of how to obtain $ \pi_{m_A}(x_A)$. Crucially, this term can be assessed through a classifier. To see the link between the definition \eqref{def_pA} and classification we introduce further notation. Let $\mathcal{S}^{P}$ be an i.i.d.~sample from $P$. We denote by $\mathcal{S}^{P}_A$ the subset of observations in $\mathcal{S}^{P}$ that is complete on A. Thus, $\mathcal{S}^{P}_A$ may be seen as an i.i.d.~draw from the density $p(\cdot|M_A=\mathbf{0})$. Similarly, we denote by $\mathcal{S}^{H}_{m_A}$ all observations that originally had missingness pattern $m_A$, but were imputed with $H$. We introduce the binary class label $Y_{m_A}$ indicating membership of an observation to the sample $\mathcal{S}^{P}_A$ by $Y_{m_A}=1$ and membership to the sample $\mathcal{S}^{H}_{m_A}$ by $Y_{m_A}=0$. We then define
\begin{align}
      \pi^{Y_{m_A}}(x_A) &:= \Prob(Y_{m_A}=1 \mid X_A=x_A), \label{probclass} \\
      \pi_{m_A} & := \Prob(Y_{m_A}=1), \notag
 \end{align} where $ \pi_{m_A} $ simply denotes the probability of class $1$ in a given projection and pattern.
 
\begin{lemma} \label{lemmadrclass}
Let $\pi_{m_A}(x_A),\pi^{Y_{m_A}}(x_A)$ and $\pi_{m_A}$ be defined as in \eqref{def_pA} and \eqref{probclass} respectively. If $\pi_{m_A}= 0.5$, then
$$\pi_{m_A}(x_A) = \pi^{Y_{m_A}}(x_A).$$
\end{lemma} 
 
The proof is given in the Appendix and simply uses Bayes Formula. 

Lemma \ref{lemmadrclass} shows that one can access the density ratio \eqref{drformula} through an estimate of the posterior probability $\pi^{Y_{m_A}}(x_A)$. This has a direct connection to classification, as for $\pi_{m_A}= 0.5$, the Bayes classifier with minimal error for this problem is given by $\Ind\{\pi^{Y_{m_A}}(x_A)>1/2\}$ \citep{ptpr}.
In practice, we ensure $\pi_{m_A}= 0.5$ in each projection $A$ and pattern $m_A$ by upsampling the smaller class to the size of the larger class, as detailed in the Appendix.




\section{Practical Aspects}\label{pracaspects}

In practice, we take several steps to estimate $S_{NA}^{*}(H,P)$ based on samples only. For a missingness pattern $m \in \mathcal{M}$, let
$\widetilde{\mathcal{A}}_m$ be a set of random projections sampled from $\mathcal{A}$. Note that we sample the set of random projections dependent on the missingness pattern $m$, as specified in Section \ref{distributionkappa}. Let furthermore $\mathcal{N}_{m_A}$ be the set of indices $i$ such that $x_{i,A}$ has missingness pattern $m_A$ whose posterior probability $\pi_{m_A}(x_A)$ in \eqref{def_pA} is estimated by $\hat{\pi}_{m_A}(x_A)$. Given an imputation method with $N \geq 1$  imputed values $x_i^1, \ldots, x_i^{N}$ of the incomplete observations, the estimator of $ S_{NA}^{*}(H,P)$ is given by
\begin{align} \label{scoreestimator}
   &\widehat S_{NA}^{*}(H,P):=\frac{1}{N}\sum_{j=1}^{N} \frac{1}{|\mathcal{M}|} \sum_{m \in \mathcal{M}}  \frac{1}{|\widetilde{\mathcal{A}}_m|}\sum_{A \in \widetilde{ \mathcal{A}}_m}  \frac{1}{|\mathcal{N}_{m_A}|}\sum_{i \in \mathcal{N}_{m_A}} \left[ \log\left( \frac{ \hat \pi_{m_A}(x_{i,A}^j)}{1- \hat \pi_{m_A}(x_{i,A}^j)} \right)\right],
\end{align}
yielding a score of the imputation performance of $H$, averaged over $N \geq 1$ imputations. 

For each projection $A$ and pattern $m_A$, we first split the data into a training and test set. We make sure to have observations with pattern $m_A$ in both the training and test set. We then fit $\hat{\pi}_{m_A}$ on the training set and evaluate it on the test set.  We use both halfs of the sample once for training and once for evaluation, to ensure that every observation contributes to the final score (\ref{scoreestimator}).

Algorithm 1 in the Appendix summarizes the practical estimation of the DR I-Score and gives more details. We now highlight a few more key concepts for the estimator.

\subsection{Random Forest Classifier}

To estimate $\pi_{m_A}(x_{i,A}^j)$ we use a classifier, as detailed in Section \ref{sec:drclass}.
Our classifier of choice is Random Forest and more specifically, the probability forest of \citet{probabilityforests}. That is, for each of the $\texttt{num.proj}$ projections in $\tilde{\mathcal{A}}_m$, we fit a Random Forest with a small number of trees (say between $5$ and $20$), a parameter called $\texttt{num.trees.per.proj}$. As such, the overall approach might be seen as one aggregated Random Forest, which restricts the variables in each tree or group of trees to a random subset of variables. This seems natural considering the construction of the RF. In each tree we set $\texttt{mtry}$ to the full dimension of the projection to avoid an additional subsampling effect. Despite the natural fit of our framework into the Random Forest construction, any other classifier may be chosen to obtain an estimate of $\pi_{m_A}$.\footnote{We also experimented with other classifiers such as generalized linear model (glm), but in practice the resulting score did not have a lot of power discriminating different imputations.}




\subsection{Distribution over Projections and Patterns} \label{distributionkappa}
The question remains how to choose the set of projections $A \sim \mathcal{K}$ from which we sample the subset $\tilde{\mathcal{A}}$ at random. We group the samples according to their missingness pattern and for each of the groups we sample \texttt{num.proj} many projections from $\mathcal{A}$, that we adapt to the given pattern.

Let $O^c_m \subseteq \{1,\ldots,d \}$ be the index set of variables with a missing value, such that $o^c(x_i,m)=x_{i,O_m^c}$ and similarly $O_m=\{1,\ldots,d \} \setminus O_m^c$ the index set of variables without a missing value. Given a missingness pattern $m$ of a group of samples, we choose $\mathcal{A}=\mathcal{A}_{m}$ as the set of subsets $A$ that satisfy $A \cap O_m^c \neq \emptyset $ and $A \setminus O_m^c \neq \emptyset $, i.e., a subset of missing indices has to be part of the projection but there must be at least one element in $A$ that is not part of $O^c_m$. The intuition behind this choice is the following: Since we want to compare a sample of projected imputed observations of a pattern to a sample of projected complete observations, we need to ensure that these samples are not the same. This is the reason for including in each projection at least one index $j$ such that $m_j=1$, i.e. $x_j$ was missing in the given pattern.


In practice, we select at random a subset $ \tilde{\mathcal{A}}_{m}$ of $\mathcal{A}_m$. To obtain each $A \in \tilde{\mathcal{A}}_{m}$ we first sample a number $r_1$ in $\{1,\ldots, |O^c_m| \}$ and a number $r_2$ in $\{1,\ldots, d-r_1 \}$. Then $A$ is obtained as the union of a random subset of size $r_1$ from $O^c_m$ and a random subset of size $r_2$ from $O_m$.

Having obtained $A \in \tilde{\mathcal{A}}_m$, we choose the support of the projected pattern to be a singleton, $\mathcal{M}_A=\{m_A\}$, with
\begin{align*}
    (m_A)_j=\begin{cases} 1, & \text{ if } j \in A\cap O_m^c \\
   0, & \text{ else. }  \end{cases}
\end{align*}
That is the pattern $m_A$ is simply the pattern $m$ projected to $A$. We thereby ensure that on the projection $A$ the training set on which the classifier is trained contains the same pattern as the test points. 

\subsection{Approximate Confidence Intervals} \label{section:jackknife}

We estimate the variance of our score, if the data would vary, by a jackknife approach as in \citet{generaljackknife}. We divide $\mathbf{X}$ randomly into two parts and compute the DR I-Score for a given imputation method for each part, obtaining $S^{(1)}$ and $S^{(2)}$. This is repeated $B$ times to obtain scores $S_1^{(1)}, \ldots, S_B^{(1)}$ and $S_1^{(2)}, \ldots, S_B^{(2)}$. Let $\bar{S}_{j}= 1/2 ( S_j^{(1)} + S_j^{(2)})$ and let $\widehat{S}$ be the score of the original data set for a given imputation method. We estimate the variance according to the formula of \citet{generaljackknife}, as
\begin{align}
    \widehat{\mbox{Var}}(\widehat{S})= \frac{1}{B} \left(  \sum_{j=1}^B \left( \bar{S}_{j} - \frac{1}{B} \sum_{j=1}^B \bar{S}_{j} \right)^2 \right). \label{formula:varest}
\end{align}

The approximate $(1-\alpha)$-Confidence Interval for our score is then given as 
\[\hat{S} \pm  q_{1-\alpha/2} \cdot \sqrt{\widehat{\mbox{Var}}(\widehat{S})},\] 
where $q_{1-\alpha/2}$ is the $(1-\alpha/2)$-quantile of a standard normal distribution. We  choose $\alpha = 0.05$ as default level. As the normality of the score is not guaranteed, a more careful approach would instead try to estimate the quantiles directly in this manner, e.g., using subagging \citep{buhlmann_yu}. While this is possible, it is computationally intense as the number of repetitions has to be high to obtain an accurate estimate of the quantiles. In contrast, the variance appears easier to approximate and simulations indicate that the estimate of the variance is reasonable. 



\section{Further Related Work}\label{relwork}

There are essentially two strands of literature related to our paper. The first concerns efforts to score imputation methods, as we do in this work. To the best of our knowledge, the research area of scoring and selecting imputation methods has not gained a lot of attention, especially under the more realistic assumption that the true data, or even a simulated ground truth, are not available for comparison. In \cite{ross} the authors state that the performance of an imputation should ``preserve the natural relationship between variables in a multivariate data set (...)". The methods they use to assess these properties include comparing densities before and after imputation and in the classification case comparing ROC curves. The authors of \cite{luengo} consider the scoring of imputation methods with respect to classification tasks. They build a score based on pairwise comparisons of classification performances of two imputation methods using Wilcoxon Tests. Again, this procedure requires the knowledge of the underlying true values.

The second line of literature concerns methods that have very different goals, but are methodologically similar to what we propose: we make use of a classifier able to handle missing values to discriminate between imputed and real data. The key idea is to use projections in the variable space together with a Random Forest (RF) \citep{Breiman2001} classifier. As such, our method is probably most closely related to the unsupervised Random Forest approach originally proposed in \citet{Breimannunsupervised} and further developed in \citet{unsupervisedRF1}. The latter uses an adversarial distribution and a RF classification to achieve a clustering in an unsupervised learning setting. We take the same approach in a different context, whereby our adversarial distribution is the imputation distribution. The classifier approach also has some connections to the famous General Adversarial Network (GAN) \citep{GAN} and the GAIN imputation method \citep{Gain} that extends the GAN approach to obtain an imputation method. Aside from the fact that instead of imputation we are concerned about evaluation, there are further differences: First, their discriminator is trained to predict the missingness pattern, while we directly compare imputed and real data. Second, compared to a fully-fledged optimized GAN, our approach based on Random Forest is simple and does not require any backpropagation or tuning. 
Finally, our approach of obtaining an estimate of the KL divergence as a ratio of estimated class probabilities was introduced in \citet{Cal2020}, who used it to construct hypothesis tests.


As we use projections in the variable space as a way of adapting a Random Forest to work with missing values, our approach is also connected to the literature of CART and Random Forest algorithms that can handle missing values. We cite and briefly summarize some of the different approaches in the literature. \cite{beaulac2018best} propose an adapted CART algorithm, called Branch Exclusive Splits Trees (BEST), where some predictors are available to split upon only within some regions of the predictor space defined according to other predictors. This structure on the variables needs to exist and be imposed by the researcher. If blind to such a structure, an easy cure is distribution-based imputation (DBI) by \cite{quinlan}. When selecting a predictor and split point, only observations with no missing value in this predictor are considered. An observation with a missing value in the variable of the split is randomized to left and right according to the distribution of the observations that have no missing value in this variable. \cite{breimanBook} proposed the approach of so-called surrogate variables. Again, only observations with no missing value are considered when choosing a predictor and split point. After choosing a best (primary) predictor and split point, a list of surrogate predictors and split points is formed. The first surrogate mimics best the split of the primary split, the second surrogate is second best and so on. This approach makes use of the correlation between the variables. There are also approaches that require fully observed training data, but are able to deal with missing values in prediction, such as \cite{saar} or \cite{nodeharvest}, but these are less relevant in our context.

\section{Empirical Validation}\label{empresults_1}
This section presents an empirical study of the performance of the DR I-Score. We do not aspire to perform an extensive comparison of state-of-the-art imputation methods, but instead employ the different imputations as a tool to validate the DR I-Score. As such, we chose commonly used imputation methods that are easily usable in \textsf{R}. We investigate whether there is evidence against propriety of the DR I-Score and assess the alignment of the ranking induced by the DR I-Score with desired distributional properties. We first list the $9$ imputation methods and $15$ real world data sets used, covering a range of numbers of observations $n$ as well as numbers of variables $d$. 

\subsection{Imputation Methods} \label{sec:imputationMethods}
We employed the following prevalent single and multiple imputation methods available in \textsf{R}, that can be divided into mice methods (``mice-cart'' , ``mice-pmm'', ``mice-midastouch'', ``mice-rf'', ``mice-norm.predict'') and others (``mipca'', ``sample'', ``missForest'', ``mean'').

All methods have in common that they are applicable to the selected data sets without indication of errors or severe warnings. For each method with prefix ``mice'' we used the \textsf{R}-package \texttt{mice} \citep{mice}. If a method required specification of parameters, we used the default values. A more detailed description of the methods used can be found in Appendix \ref{impmethodsdetail}.

\subsection{Data Sets}
We used the real data sets specified in Table \ref{Tab:data sets} for the assessment of the DR I-Score. They are available in the UCI machine learning repository\footnote{\url{https://archive.ics.uci.edu/ml/index.php}}, except for Boston (accessible in \textsf{R}-package \texttt{MASS}) and CASchools (accessible in \textsf{R}-package \texttt{AER}). We preprocessed the data by cancelling factor variables, in order to be able to run all the assessed imputation methods without errors. However, we kept numerical variables with only few unique values. This preprocessing was done solely for the imputation methods, our proposed score could be used with factor variables as well. Finally, in the data set ecoli we deleted two variables because of multicollinearity issues. 


\begin{table}
\begin{center}
\begin{tabular}{|c|c c|}
\hline
data set            & n & d \\[0.5ex] 
\hline \hline
airfoil                     & 1503       & 6                \\
Boston                      & 506        & 14            \\
CASchools                   & 420        & 10              \\
climate.model.crashes     & 540        & 19                  \\
concrete.compression       & 1030       & 9                   \\
concrete.slump             & 103        & 10                \\
connectionist.bench.vowel & 990        & 10                 \\
ecoli                       & 336        & 5           \\
ionosphere                  & 351        & 32                 \\
iris                        & 150        & 4                 \\
planning.relax             & 182        & 12                \\
seeds                       & 210        & 7                         \\
wine                        & 178        & 13            \\
yacht                       & 308        & 7               \\
yeast                       & 1484       & 8       \\
\hline
\end{tabular}
\end{center}
\caption{Data sets used for performance assessment of the DR I-Score with number of observations $n$ and number of variables $d$.}
\label{Tab:data sets}
\end{table}

\subsection{Propriety of DR I-Score} \label{sec:propriety}


In this section we check empirically whether any imputation is ranked significantly higher than the true underlying data. This is an attempt to empirically assess if the DR I-Score is proper and a lack of significance might also indicate an insufficient sample size to detect a violation. In addition, we can check how well each method performed on each data set with respect to the DR I-Score. To test empirical propriety of the score, that is the non-inferiority of the true data score, we score the fully observed data set by $ \widehat S^{*}_{NA}(P^*,P)$ and the imputed data set by $\widehat S^{*}_{NA}(H,P)$ for each imputation distribution $H$ and consider the difference $D_{H} := \widehat S^{*}_{NA}(H,P) -  \widehat S^{*}_{NA}(P^*,P)$.
We want to test the following hypotheses for all $H$:
\begin{align}
    H_0: D_H = 0 \ \text{vs} \ H_A: D_H > 0. \label{test1}
\end{align}
We do this by $p$-value calculation assuming that approximately
\begin{equation}
    D_H \overset{H_0}{\sim} \mathcal{N}(0, \sigma^2(D_H)), \label{ass:distrH0}
\end{equation}
where we estimate $\sigma(D_H)$ with the Jackknife variance estimator formula (\ref{formula:varest}) using $30$ times $1/2$-subsampling. Details of the whole simulation are given in the Appendix.

We fixed the overall fraction of missingness to $p_{miss}=0.2$ (Results for $p_{\text{miss}}= 0.1$ can be found in the Appendix). When setting data to \texttt{NA}, we applied the MCAR and the MAR mechanism: In MCAR, we set each value in the data set to missing with probability $p_{miss}$. In MAR we created $d/2$ random missingness patterns $m$ for a data set of dimension $d$. Afterwards, we used the ``ampute'' function of the package \texttt{mice} with these patterns, where all patterns have the same frequency, to create a MAR data set. In the MAR case, $p_{\text{miss}}$ might slightly deviate from $0.2$. 

As parameters for the DR I-Score estimation we chose the number of trees per projection (\texttt{num.trees.per.proj}) to be $5$  and the minimal node size in each tree to be $10$ (the default for a probability RF). We chose the number of projections (\texttt{num.proj}) adaptively to the dimension $d$ of the data set: for $d \leq 6$ we used $50$, for $7 \leq d \leq 14$ we used $100$ and for $d \geq 15$ we used $200$. We set the number of imputations to $m=5$ (the default value in the \textsf{R}-package \texttt{mice}). 

We generated a realization of the \texttt{NA}-mask for the MCAR and the MAR case and then computed for each method/data set combination the corresponding $p$-value of testing \eqref{test1} under assumption \eqref{ass:distrH0}. All methods were computable on all data sets, without throwing errors or major warnings, except for mice-midastouch on yeast, indicated with an \texttt{NA} in Figure \ref{fig:pvals}. Our findings for testing propriety can be summarized as follows: At level $\alpha = 0.05$ we found no single significant $p$-value in the MCAR or in the MAR case. At level $\alpha = 0.1$ we found in the MAR case two significant $p$-values for mice-rf in the data sets yeast and concrete.slump and in the MCAR case one significant $p$-value for mice-cart in ionosphere. The latter data set has the highest dimension of $32$, which for MCAR leads to the extreme case of only observing each missingness pattern once. In summary, these results do not reveal enough evidence against the propriety of our estimated DR I-Score. We point out that, in the MAR case, we did not verify the MAR assumption on the projections.
\begin{figure}
     \centering
     \begin{subfigure}[b]{0.48\textwidth}
         \centering
         \includegraphics[width=1\textwidth]{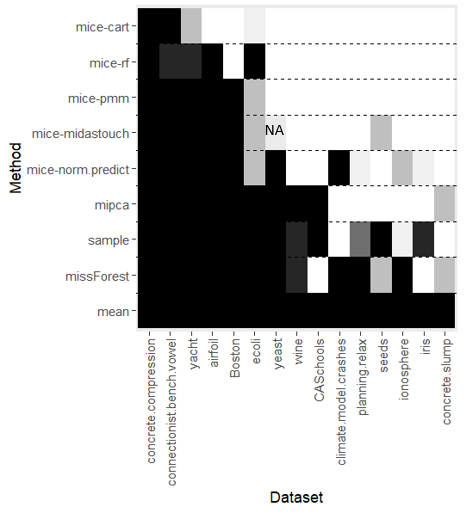}
              \caption{MAR}
     \end{subfigure}
     \begin{subfigure}[b]{0.48\textwidth}
         \centering
           \includegraphics[width=1.2\textwidth]{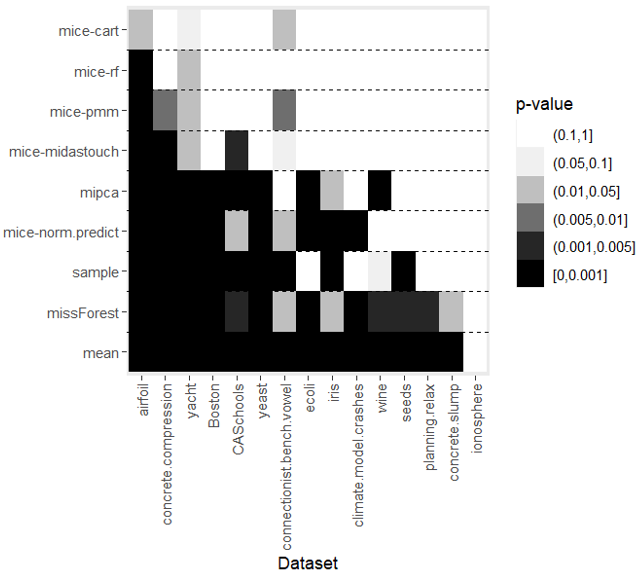}
         \caption{MCAR}
     \end{subfigure}
   \caption{Discretized $p$-values of testing \eqref{test2} under assumption \eqref{ass:distrH0} for the $9$ methods applied to the $15$ data sets. We used missingness mechanisms MAR (a) and MCAR (b), $p_{miss}=0.2$ and $m=5$. The parameter values of the DR I-Score are described in the text. ``NA'' means that the method was not computable on the respective data set.}
\label{fig:pvals}

\end{figure}

Reversing the alternative hypothesis we obtain,
\begin{align}
    H_0: D_H = 0 \ \text{vs} \ H_A: D_H < 0, \label{test2}
\end{align}
whose corresponding test may reveal more information about the performance of the methods as well as the difficulty for imputation of each data set. We assume \eqref{ass:distrH0} and calculate the corresponding $p$-values for \eqref{test2}. In Figure~\ref{fig:pvals}, we discretized each $p$-value corresponding to a method applied to a data set into one of $6$ batches, which reflect different significance levels. The larger the $p$-value, the lighter the shade, the better the respective method imputed on the given data set with respect to our score. With a $p$-value in the batch $(0.1, 1]$ (white) we can not reject the null that $D_H=0$ even at the level $0.1$, i.e., the imputation performed as well as the true data. We sorted the rows and columns to cluster similarly scored methods and data sets together.

Reading the plot row-wise reveals performance information of the methods, the higher up a method appears the better it performed according to the DR I-Score. Interestingly, the MAR and MCAR case reveal almost exactly the same ordering of the methods: Only ranks of mice-norm.predict and mipca are flipped, however these two reveal very similar $p$-value patterns in both plots. If we divide the methods into two groups, we observe in the better group the mice methods (mice-cart, mice-rf, mice-pmm, mice-midastouch). The best method overall in both cases is mice-cart, whose imputations were indistinguishable from the true data in $11$ data sets in the MAR case and $12$ in the MCAR case, even at level $0.1$. In addition, we want to emphasize the suboptimal performance of methods that predict conditional means, without additional randomization to impute values, in particular missForest and mean. This may be surprising as missForest is known to perform very well in the literature, see e.g. \citet{Waljee2013}. However this impression of good performance is based on measures of accurary, such as RMSE. As laid out in \citet{VANBUUREN2018}, as a prediction method, missForest does not account for the uncertainty caused by the missing data. Contrary to accuracy measures, our score takes the joint distribution into account when assessing performance, hence the comparatively weak performance of missForest over the chosen data sets. The worst method appears to be the mean, a method known to heavily distort the distribution.

Reading the plot column-wise reveals information about how easy or hard the data sets could be imputed by the methods considered. The further to the right a data set appears, the easier it was to impute. For instance, we find that none of the considered methods was able to find an imputation that recovers the joint distribution well enough for the data set concrete.compression (MAR) and airfoil (MCAR). 


\subsection{Relevancy of DR I-Score} \label{sec:relevancy}

In the last section, we did not discover evidence against the propriety of the estimated DR I-Score. However, a practitioner might want to select the best method with the score, or at least determine the worst performing methods to definitely not employ them. Unfortunately, there is no ground truth to compare to. However, one would hope that the methods chosen by our score perform well on a wide range of targets, even though it was not designed to select for any of these targets specifically. We focus on a target that presumably will be of interest in practice when doing multiple imputation: Average coverage and average width of marginal confidence intervals for each \texttt{NA} value, obtained by the $m$ multiple imputations. More specifically, let $\mathcal{N} \subseteq \{1,\ldots, n\}$ be the set of $i$ for which $\sum_{j=1}^p \mathbf{M}_{ij} >0$, i.e., those observations that contain at least one missing value. For each observation $x_i$, let $\mathcal{N}_i \subset \{1,\ldots, p\}$ be the missing coordinates of $x_i$. Given an imputation method $H$, we obtain for each $x_{ij}$, $i \in \mathcal{N}$, $j \in \mathcal{N}_i$, $m$  imputed values $x_{ij}^1(H), \ldots, x_{ij}^m(H)$. For $m$ large enough, we compute the empirical $0.025$- and $0.975$-quantiles $x_{ij}^{0.025}(H)$ and $x_{ij}^{0.975}(H)$ and consider the interval spanned in between as the empirical $95$\%-CI for $x_{ij}$ defined by the method $H$ through $x_{ij}^1(H), \ldots, x_{ij}^m(H)$. We denote it by 
\begin{equation*}
    \text{CI}( x_{ij}^1(H), \ldots, x_{ij}^m(H)):= [x_{ij}^{0.025}(H), x_{ij}^{0.975}(H)].
\end{equation*}
We then check whether the true missing data point $x_{ij}$ lies within the CI or not and obtain an average marginal coverage for the method $H$ by averaging over all $i \in \mathcal{N}$, i.e., 
\begin{equation*}
    \text{Coverage}(H):= \frac{1}{\mathcal{N}}\sum_{i\in \mathcal{N}}\frac{1}{\mathcal{N}_i}\sum_{j\in \mathcal{N}_i}\Ind\{x_{ij} \in \text{CI}( x_{ij}^1(H), \ldots, x_{ij}^m(H))\}.
\end{equation*}

\begin{figure}
     \centering
     \begin{subfigure}[b]{0.49\textwidth}
         \centering
         \includegraphics[width=1.1\textwidth]{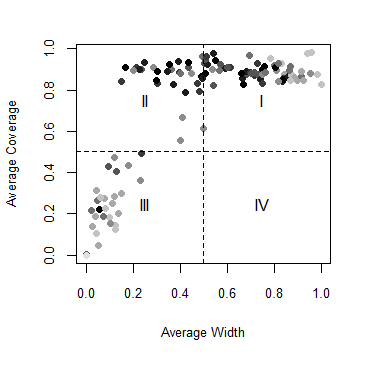}
         \vspace*{-1.2cm}
              \caption{MAR}
     \end{subfigure}
     \begin{subfigure}[b]{0.49\textwidth}
         \centering
           \includegraphics[width=1.1\textwidth]{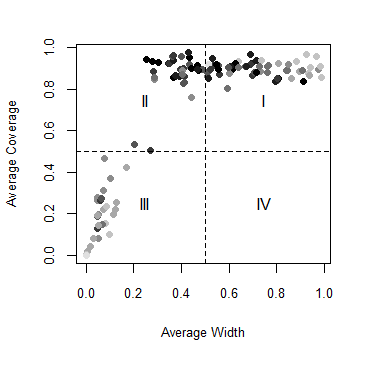}
           \vspace*{-1.2cm}
         \caption{MCAR}
     \end{subfigure}
 \caption{Average coverage plotted against average width for the $9$ methods applied to the $15$ data sets (total = $9\times15=135$ points). The darkness indicates the rank induced by the DR I-Score (the darker, the higher the rank). We used the missingness mechanism MAR in (a) and MCAR in (b) with $p_{miss}=0.2$, $m=20$ and the in the text described parameter values to compute the DR I-Scores.}
  \label{fig:cov}
\end{figure}

A method that has a large enough variation between the $m$ different imputations will reach a coverage of $1$, however the average width of its corresponding CI is large, indicating very little power. We use the average width of the confidence intervals as an indicator of statistical efficiency. A good imputation method will have small average width while still maintaining high average coverage. We define the average marginal width for the method $H$ by
\begin{align*}
     \text{Width}(H):= \frac{1}{\mathcal{N}}\sum_{i\in \mathcal{N}}\frac{1}{\mathcal{N}_i}\sum_{j\in \mathcal{N}_i} \left( x_{ij}^{0.975}(H)-x_{ij}^{0.025}(H)\right).
\end{align*}
For better visualization we constrain the $\text{Width}(H)$ for all methods $H$ to the interval $[0,1]$: Given a data set, we normalize $\text{Width}(H)$ for all methods $H$ by the maximal width of all methods. 


In Figure~\ref{fig:cov} we plot $\text{Coverage}(H)$ against the normalized $\text{Width}(H)$ for all methods applied to all data sets, leading to totally $15\times 9=135$ points. Not all of them are visible since they can lie on top of each other. For example the method mean always produces the point $(0,0)$ for all $15$ data sets, since there is no variation in the $m$ imputed data sets. The shade of the points reflects the induced rank by the DR I-Score, assigning one of the $9$ ranks to each method in a given data set: rank $1$ to the best scored method (black) up to rank $9$ to the worst scored method (light gray). We set the number of imputations to $m=20$, used missingness mechanism MCAR and MAR and $p_{\text{miss}}=0.2$ and the same parameter setting as in Section \ref{sec:propriety}. For easier interpretation, we name the four quadrants of the square in Figure~\ref{fig:cov} by the letters I-IV. We observe that quadrant IV, corresponding to high average width and low average coverage, does not contain any points, which makes sense. Considering the shade gradient, the DR I-Score seems to indeed often rank points with high average coverage and low average width the best. Clearly, methods that produce small average width combined with small average coverage as well as high average width combined with high average coverage, are ranked low, which is highly desirable.

 \begin{figure}
\begin{minipage}{0.45\textwidth}
\centering
\resizebox{1\textwidth}{!}{
\begin{tabular}{|c|c c  c|}
  \hline
method/quadrant & I    & II    & III \\[0.5ex] 
\hline\hline
cart & 0.47 & 0.53 & 0 \\ 
  pmm & 0.53 & 0.47 & 0 \\ 
  midastouch & 0.67 & 0.33 & 0 \\ 
  rf & 0.73 & 0.27 & 0 \\ 
  mipca & 0.87 & 0.13 & 0 \\ 
  sample & 0.87 & 0.13 & 0 \\ 
  norm.predict & 0 & 0.13 & 0.87 \\ 
  mean & 0 & 0 & 1 \\ 
  missForest & 0 & 0 & 1 \\ 
   \hline
\end{tabular}
}
\subcaption{MAR}
\end{minipage}
\begin{minipage}{0.45\textwidth}
 \resizebox{1\textwidth}{!}{
\begin{tabular}{|c|c c  c|}
  \hline
method/quadrant & I    & II    & III \\[0.5ex] 
\hline\hline
pmm & 0.40 & 0.60 & 0 \\ 
  cart & 0.47 & 0.53 & 0 \\ 
  midastouch & 0.53 & 0.47 & 0 \\ 
  rf & 0.60 & 0.40 & 0 \\ 
  mipca & 0.87 & 0.13 & 0 \\ 
  norm.predict & 0 & 0.13 & 0.87 \\ 
  sample & 0.93 & 0.07 & 0 \\ 
  mean & 0 & 0 & 1 \\ 
  missForest & 0 & 0 & 1 \\ 
   \hline
\end{tabular}
}
\subcaption{MCAR}
\end{minipage}
\captionof{table}{The fraction of times each method appeared in the quadrants I, II and III of Figure \ref{fig:cov} in the MAR case (a) and the MCAR case (b).}
\label{tab:resmethods}
\end{figure}


Table \ref{tab:resmethods} shows a clearer picture of the roles of the methods in these results. We indicate for each method the fraction of times in the data sets it appeared in the quadrants I, II and III in Figure \ref{fig:cov}. We observe that all points making up quadrant III are from the methods mice-norm.predict, mean and missForest, i.e., the methods that use a sort of mean imputation. From a distributional point of view, these methods are often not favourable, which is a property nicely captured by the DR I-Score.

\subsection{Real Data Example: Births Data}
We further illustrate our method on the natality data of $2020$ obtained from the Centers for Disease Control and Prevention (CDC) website.\footnote{\url{https://www.cdc.gov/nchs/data_access/vitalstatsonline.htm}} This data set contains information on $\approx 3.5$ million births in $2020$ in the US. We subsample the data as detailed later and consider $23$ variables, listed in Appendix \ref{app:birth}. The variables include categorical variables, such as race and education of mother and father and the gender of the newborn, as well as continuous variables, such as the age of the parents and the baby's birth weight. 
We consider a subset of imputation methods that can deal with these mixed data: mice-pmm, mice-cart, sample and missForest. We also include the mean as a baseline comparison.

\begin{figure}
     \centering
     \begin{subfigure}[b]{0.48\textwidth}
         \centering
         \includegraphics[width=1.1\textwidth]{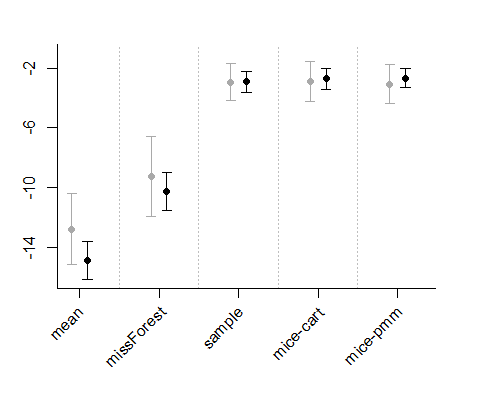}
         \vspace*{-1cm}
              \caption{Estimated DR I-Score with CIs.}
     \end{subfigure}
     \begin{subfigure}[b]{0.48\textwidth}
         \centering
         \includegraphics[width=1.1\textwidth]{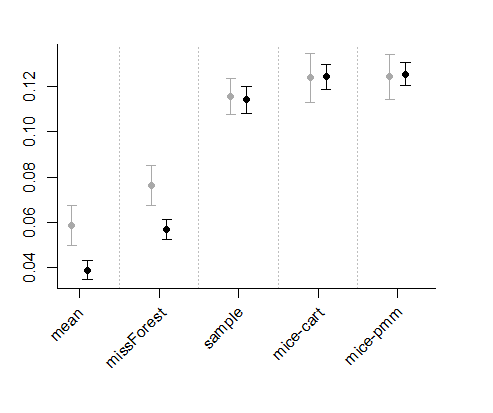}
         \vspace*{-1cm}
         \caption{OOB error with CIs.}
     \end{subfigure}
     \caption{Births data: Estimation of the proposed DR I-Score (a) and the OOB error of the classification RF (b) with corresponding CIs for the methods mean, missForest, sample, mice-cart and mice-pmm with sample size $n=500$ (grey) and $n=2000$ (black). We obtained the CIs by subsampling as described in Section \ref{section:jackknife}.}
\label{fig:birth}
     \end{figure}

We subsample the data to obtain two smaller subsets, one of size $500$ and one of size $2000$, where we sample at random observations with at least one missing value. As such, we obtain in both samples an overall probability of missingness of $p_{miss}\approx 0.08$, which is lower than in the previous sections. Missing values are encountered predominantly in the variables weight of mother, weight gain of mother and BMI of mother, which are correlated, and variables about the father, such as the age, race or education. We estimate the DR I-Score for the aforementioned imputation methods applied to the two data sets (sizes $500$, $2000$) with missing values.

Since the births data set also contains a lot of complete cases, it is possible to validate our score in a special way: We compare each of the imputations with a disjoint sample of only complete cases of the same size that we obtained by randomly sampling complete cases from the whole data set. The comparison is assessed via a classification Random Forest, distinguishing between the sample of complete cases and the imputation, where we report the out-of-bag prediction error (OOB error). We note that in this case, the smaller the OOB error, the easier it was for the classifier to distinguish the imputation from the true data. Hence, the smaller the OOB error, the worse the imputation method. For both the score and the OOB error we also compute $95\%$-CI with the Jackknife variance estimator formula (\ref{formula:varest}) using $30$ times $1/2$-subsampling. We plotted the results in Figure \ref{fig:birth} to compare the ranking of methods obtained by our score with the one obtained by the RF. We ordered the methods according to the mean score/OOB error, computed on the sample of size $2000$.

First, we observe that the rankings obtained from the DR I-Score and the OOB error of the RF are the same. Second, the ranking is in line with the one we obtained in Section \ref{sec:propriety}: mean and missForest appear to be the weakest methods followed by sample, while mice-cart and mice-pmm are ranked best. Third, by increasing the sample size from $500$ (gray) to $2000$ (black), the CIs get more narrow and the methods can be held apart more significantly for both the score and the RF. We observe that the score cannot distinguish sample from the two best methods as clearly as the RF, which would be desirable. However, the RF is allowed to use an additional (large) complete sample, which is naturally an unfair advantage. Moreover, this method is generally not applicable without this large amount of fully observed points. In contrast, the DR I-Score works with only incomplete observations, and still manages to reproduce the same ordering as the RF. 


\section{Discussion}\label{sec:discussion}
In this paper we presented the convenient framework of Imputation Scores to score a given set of imputation methods for an incomplete data set. The widespread assessment of imputation methods via RMSE (of masked observations) favors methods that impute conditional means but do not necessarily reflect the whole conditional distributions. Given the assumption of MAR on each projection, our proposed density ratio I-Score is able to give high scores to methods that replicate the data distribution well.

\newpage
\begin{appendix}

%

\section{Details of Imputation Methods} \label{impmethodsdetail}

\begin{itemize}
    \item[1)] \textbf{missForest} is a multiple imputation method based on iterative use of RF, allowing for continuous and categorical data \citep{stekhoven_missoforest}. After an initial mean-imputation, the variables are sorted according to their amount of missing values, starting with the lowest. For each variable as response, a RF is fitted based on the observed values. The missing values are then predicted with the RF. The imputation procedure is repeated until a stopping criterion is met. We used the \textsf{R}-package \texttt{missForest} \citep{missforest_package}.
    \item[2)] \textbf{mipca} is a multiple imputation method with a PCA model \citep{josse_mipca}. After an initialization step, an EM algorithm with parametric bootstrap is applied to iteratively update the PCA-parameter estimates and draw imputations from the predictive distribution. The algorithm is implemented in the function \texttt{MIPCA} of the \textsf{R}-package \texttt{missMDA} \citep{mipca_package}. We use the function \texttt{estim\_ncpPCA} to estimate the number of dimensions for the principal component analysis by cross-validation.
    \item[3)] \textbf{mean} is the simplest single imputation method considered. It imputes with the mean of the observed cases for numerical predictors and the mode of observed cases for categorical predictors. We use the implementation of the \textsf{R}-package \texttt{mice}.
     \item[4)] \textbf{sample} is a multiple imputation method sampling at random a value of the observed observations in each variable to impute missing values. We use the implementation of the \textsf{R}-package \texttt{mice}.
    \item[5)] \textbf{mice-cart} is a multiple imputation method cycling through the following steps multiple times \citep{mice-rf-method}: After an initial imputation through sampling of the observed values, a classification or regression tree is fitted. For each observation with missing values, the terminal node they end up according to the fitted tree is determined. A random member in this node is selected of which the observed value leads the imputation. 
    \item[6)] \textbf{mice-norm.predict} is a multiple imputation method cycling through the same steps as mice-cart with the adaptation that a linear regression is fitted and its predicted value is used as imputation. 
    \item[7)] \textbf{mice-pmm} Predictive Mean-Matching is a semi-parametric imputation approach (\cite{little_pmm} and \cite{rubin_pmm}). Based on the complete data, a linear regression model is estimated, followed by a parameter update step. Each missing value is filled with the observed value of a donor that is randomly selected among complete observations being close in predicted values to the predicted value of the case containing the missing value. 
    \item[8)] \textbf{mice-midastouch} is a multiple imputation method using an adaption of classical predictive mean-matching, where candidate donors have different probabilities to be drawn \citep{midastouch}. The probability depends on the distance between the donor and the incomplete observation. A closeness parameter is specified adaptively to the data. 
    \item[9)] \textbf{mice-rf} is a multiple imputation method cycling through the same steps as mice-cart with the adaptations that one tree is fitted for every bootstrap sample. For each observation with missing values, the terminal nodes in each tree are determined. A random member of the union of the terminal nodes is selected of which the observed value leads the imputation.
\end{itemize}

\section{Details on Birth Data} \label{app:birth}

Even though the original data contains a lot of variables, we took only the following variables from the source data:
\begin{itemize}
    \item mother's age, height, weight before the pregnancy, weight gain during pregnancy and BMI before pregnancy
    \item mother's race (black, white, asian, NHOPI, AIAN or mixed), marital status (married or unmarried) and the level of education (in total $8$ levels)
    \item father's age, race and level of education
    \item month of birth
    \item plurality of the birth (how many babies were born at once)
    \item whether and when the prenatal care started
    \item pregnancy duration
    \item delivery method (vaginal or C-section)
    \item birth order - the total number of babies born by the same mother (including the current one)
    \item birth interval - number of months passed since last birth (NA if this is the first child)
    \item number of cigarettes smoked per day on average before and during the pregnancy
    \item birth weight (in grams) and gender of the baby
    \item indicators whether baby had any abnormal condition.
\end{itemize}

\section{Algorithm and Implementation Details of DR I-Score Estimation}\label{Implementationdetails}

\begin{algorithm}
\normalsize
\textbf{Inputs}: data set $X$ containing missing values to impute, a multiple imputation method applied to $X$ yielding $N$ imputed data sets $\hat{X}_i$, i=1,\ldots, N\;
\KwResult{DR I-Score $s$ for the imputation method as the average of the $N$ scores $\{s_i\}_{i=1}^{N}$ for each of the $N$ imputed data sets}
 \textbf{Hyper-parameters}: number of projections \texttt{num.proj}, number of trees per projection \texttt{num.trees.per.proj}, standard parameters of the Probability Forests\;
 - Group observations in $X$ into $J$ different groups according to their unique missing value patterns $M_j$, $j=1,\ldots, J$.\;
 \For{$i=1,\ldots,N$}{
   \For{$m=M_1,\ldots,M_J $}{
   - Sample a set of $\texttt{num.proj}$ projections, $A_k$, $k=1,\ldots, \texttt{num.proj}$ compatible with the missingness pattern $m$ as described in Section 5.2\;
   - Get the projected imputed data $\hat{X}_i$ with pattern $m$, and split them in two halves $\hat{X}_i^{0}$ and $\hat{X}_i^{1}$\;
    \For{$l=0,1$}{
    \For{$k=1,\ldots, \texttt{num.proj}$}{
    - Get the complete observations $X^{comp}_{A_k}$ from the projected data $X_{ A_k}$\;
   - Get the projected imputed data $\hat{X}^{l}_{i,A_k}$\;
    - Fit a Probability Forest with $\texttt{num.trees.per.proj}$ trees and $\texttt{mtry}$ full, discriminating $X^{comp}_{A_k}$ from $\hat{X}^{l}_{i,A_k}$ (ensuring a balance of classes, see above for details)\;
   }
   - Form one Probability Random Forest based on the \texttt{num.proj} many random forests\;
  - Get an estimate of the density ratio, $\frac{\hat{\pi}_{m_A}}{1-\hat{\pi}_{m_A}}$ , through this Probability Forest as in Equation (4.7)\;  
   - Compute the individual score contributions, as $\log \frac{\hat{\pi}_{m_A}(x)}{1-\hat{\pi}_{m_A}(x)} $ of the left-out imputations $ x \in \hat{X}^{1-l}_{ i,A_k}$\; 
 }
   - Average for $l=0,1$ the individual score contributions of all points, leading to score $s_{i,m}$\;
 }
 - Average the scores $s_{i,m}$ over all patterns $m$ to get score $s_i$\;
 }
 - Average the score $s_{i}$ over all imputations $i$ to get the final score $s$.
 
 \caption{Algorithm for estimation of DR I-Score}
\end{algorithm}

Here we present additional details of the implementation of the DR I-Score and the full algorithm in pseudo-code.

\hspace*{2cm}

\noindent \textbf{Projection Distribution.} In Section $5.2$ we describe the distribution $\mathcal{K}$ over the projections with restricted support used for the empirical estimation of the DR I-Score. The choice of this distribution is up to the experimenter and can be adapted to the specific patterns of missinginess in a given data set. In particular, the projections can be chosen such that each projection satisfies the MAR assumption, if this can be determined with domain knowledge. \\

\noindent
\textbf{Ensuring Class Balancing. }
We follow a simple procedure to ensure the same number of observations in the training sets $\mathcal{S}_A^P$ and $\mathcal{S}_{m_A}^H$. First if $\mathcal{S}_{m_A}^H$ has fewer elements than $\mathcal{S}_A^P$, but is ``large enough'' relative to $\mathcal{S}_A^P$, we simply upsample $\mathcal{S}_{m_A}^H$ with replacement until it contains the same number of elements as $\mathcal{S}_A^P$. The exact same procedure is applied if $\mathcal{S}_A^P$ has fewer elements than $\mathcal{S}_{m_A}^H$. On the other hand, if the set $\mathcal{S}_{m_A}^H$ is smaller than $\mathcal{S}_{A}^P$, that is if $|\mathcal{S}_{m_A}^H| < \tau |\mathcal{S}_{A}^P|$ for some $\tau \in (0, 1)$, we sample with replacement
observations from other patterns and add them to $\mathcal{S}_{m_A}^H$. We found that $\tau = 0.75$ works well empirically. This is done to ensure that we do not upsample one or two observations. In practice it seems adding additional patterns in the training step of the classifier does not hurt propriety.\\

\noindent    
\textbf{Numerical Truncation.} To avoid numerical issues when calculating the log of the density ratio with Expression (4.7), we apply for each $A$, $m_A$ and $x_A$ the following truncation function to $\hat \pi_{m_A}(x_A)$ 
    \begin{equation*}
        p(x) = \min(\max(x, 10^{-9}), 1-10^{-9}).
    \end{equation*}
Thus, we slightly adapt the predicted probabilities $\hat \pi_{m_A}(x_A)$ that are close to $0$ or $1$, such that the log of (4.7) can be computed. \\

\noindent    
\textbf{Patterns of Size $\mathbf{1}$.} In our algorithm to estimate the DR I-Score, we split the observations of a given missingness pattern $m$ into a training and a test set. We then use the training set to estimate the RF and use it to predict on the test set. Predominantly in the MCAR case, but also potentially in MAR situations, it often happens that several patterns contain only one observation. In this case, we group these patterns of size $1$ together, use all of these samples for training and testing and only fit one RF based on several projections for this new group. This is mainly done due to computational reasons.


\section{Proofs}\label{proofssec}

\begin{proposition}[Restatement of Proposition $4.1$]
Let $H\in \mathcal{H}_{P}$, as defined in (4.1). If for all $A\in \mathcal{A}$,
    \begin{align*}
       &p^*(o^c(x_A,m_A)|o(x_A,m_A), M=m'_A) = p^*(o^c(x_A,m_A)|o(x_A,m_A)), \nonumber \\
       &\text{ for all } m'_A, m_A \in \mathcal{M}_A \tag{4.4},
    \end{align*}
 then $S^{*}_{NA}(H, P)$ in (4.3) is a proper I-Score.
\end{proposition}


\begin{proof}
Since,
\begin{align*}
S_{NA}^{*}(H,P) = - \E_{A \sim  \mathcal{K}, M_A \sim P_A^M}D_{KL}(h_{M_A} \mid \mid p_A(\cdot \mid M_A=\mathbf{0})),
\end{align*}
it is enough to show that for all $A \in \mathcal{A}$, and all $m_A \in \mathcal{M}_A$,
\begin{align*}
D_{KL}(h_{m_A} \mid \mid  p_A( \cdot | M_A=\mathbf{0})) \geq  D_{KL}(p^*_{m_A} \mid \mid  p_A( \cdot| M_A=\mathbf{0})).
\end{align*}
We thus drop the subscript $A$ for simplicity. 

It then holds that for all $m \in \mathcal{M}$,
\begin{align*}
 h_m(x)= h(x|M=m)&=h(o^c(x, m)| o(x, m), M=m) p^*(o(x,m)| M=m),
\end{align*}
by the definition of $\mathcal{H}_P$. Similarly, 
\begin{align*}
    p(x|M=\mathbf{0})&=p^*(o^c(x,m)| o(x,m), M=\mathbf{0}) p^*(o(x,m) |  M=\mathbf{0}) \\
    &=p^*(o^c(x,m)| o(x,m)) p^*(o(x,m) |  M=\mathbf{0}),
\end{align*}
where we used assumption (4.4) in the last step. Crucially, we can then decompose the KL divergence:
\begin{align*}
    &D_{KL}(h_{m} \mid \mid  p( \cdot | M=\mathbf{0}))\\
    &= \int \log \left( \frac{h(o^c(x, m)| o(x, m), M=m) p^*(o(x,m)| M=m)}{p^*(o^c(x,m)| o(x,m))  p^*(o(x,m) |  M=\mathbf{0})} \right) h(x \mid M=m) d\mu(x) \\
        &= \int \log \left( \frac{h(o^c(x, m)| o(x, m), M=m)}{p^*(o^c(x,m)| o(x,m))} \right) h(o^c(x, m)| o(x, m), M=m) p^*(o(x,m)| M=m) d\mu(x)\\
        &+\int \log \left( \frac{ p^*(o(x,m)| M=m)}{ p^*(o(x,m) |  M=\mathbf{0})} \right) h(o^c(x, m)| o(x, m), M=m) p^*(o(x,m)| M=m) d\mu(x)\\
         &=\E_{o(x, m) \sim p^*_m} \left[ \int \log \left( \frac{h(o^c(x, m)| o(x, m), M=m)}{p^*(o^c(x,m)| o(x,m))} \right) h(o^c(x, m)| o(x, m), M=m) d\mu(o^c(x,m)) \right]\\
        &+\int \log \left( \frac{ p^*(o(x,m)| M=m)}{ p^*(o(x,m) |  M=\mathbf{0})} \right)  p^*(o(x,m)| M=m) d\mu(o(x,m)).\\
\end{align*}
The second summand in the last term is simply the KL divergence between $p^*(o(x,m)| M=\mathbf{0})$ and $p^*(o(x,m)| M=m)$ and represents the irreducible part, as it cannot be changed by any imputation. The first summand is bounded below by zero, and attains zero for $h=p^*$, proving the claim.
\end{proof}

\citet{whatismeant} point out the problematic definition of MAR throughout the literature and define MAR as:
\begin{align}
&\Prob(M_A=m_A | X_A=x_A) = \Prob(M_A=m_A| X_A=\tilde{x}_A) \nonumber \\
& \text{ for all } x_A, \tilde{x}_A \text{ s.t. } o(x_A,m_A)=o(\tilde{x}_A,m_A) \tag{S2.1}.
\label{MARcondtrue}
\end{align}
In the following we ensure that (4.4) and \eqref{MARcondtrue} really mean the same: 

\begin{lemma}
Condition (4.4) and \eqref{MARcondtrue} are equivalent.
\end{lemma}

\begin{proof}
Throughout we drop the subscript $A$ for simplicity. We start by reformulating \eqref{MARcondtrue}, for any $x, \tilde{x}$ such that $o(x,m)=o(\tilde{x},m)$,
\begin{align}\label{MARCondtrue2}
    &\Prob(M=m | X=x) = \Prob(M=m| X=\tilde{x}) \Leftrightarrow \nonumber \\
    & \frac{p^*(x  |M =m )\Prob(M =m )}{p^*(x )} = \frac{p^*(\tilde{x}  |M =m )\Prob(M =m )}{{p^*(\tilde{x} )}} \Leftrightarrow \nonumber\\
    & \frac{p^*(o(x ,m ), o^c(x ,m ) \mid M  =m )}{p^*(o(\tilde{x} ,m ), o^c(\tilde{x} ,m ) \mid M  =m )} = 
    \frac{p^*(o(x , m ), o^c(x ,m ))}{p^*(o(\tilde{x} , m ), o^c(\tilde{x} ,m ))} \Leftrightarrow  \nonumber\\
    & \frac{p^*(o^c(x ,m ) \mid o(x ,m ), M =m )}{p^*(o^c(x ,m )\mid o(x ,m ))} = \frac{p^*(o^c(\tilde{x} ,m ) \mid o(x ,m ), M =m )}{p^*(o^c(\tilde{x} ,m )\mid o(x ,m ))}  \Leftrightarrow \nonumber\\
    & p^*(o^c(x ,m ) \mid o(x ,m ), M =m ) = \frac{p^*(o^c(\tilde{x} ,m ) \mid o(x ,m ), M =m )}{p^*(o^c(\tilde{x} ,m )\mid o(x ,m ))} p^*(o^c(x ,m )\mid o(x ,m )) \tag{S2.2}
\end{align}
Clearly, (4.4) implies \eqref{MARCondtrue2}. Integrating \eqref{MARCondtrue2} with respect to the missing part of $x$, $o^c(x,m)$, only shows that 
\[
\frac{p^*(o^c(\tilde{x} ,m ) \mid o(x ,m ), M =m )}{p^*(o^c(\tilde{x} ,m )\mid o(x ,m ))}=1,
\]
and thus also (4.4).
\end{proof}

\begin{lemma}[Restatement of Lemma 4.1]
Let $\pi_{m_A}(x_A),\pi^{Y_{m_A}}(x_A)$ and $\pi_{m_A}$ be defined as in (4.6) and (4.8) respectively. If $\pi_{m_A}= 0.5$, it holds that
$$\pi_{m_A}(x_A) = \pi^{Y_{m_A}}(x_A).$$
\end{lemma} 
 
\begin{proof}
Given the definition of the class label $Y_{m_A}$, we can rewrite $p_A(x_A \mid M_A=\mathbf{0})$ and $ h_{m_A}(x_A)$ by
\begin{align*}
  p_A(x_A \mid M_A=\mathbf{0}) &= f(x_A \mid  Y_{m_A}=1),\\
 h_{m_A}(x_A) &= f(x_A \mid  Y_{m_A}=0).
\end{align*}  
By Bayes Formula it now follows
\begin{align*}
       \pi^{Y_{m_A}}(x_A) &:= \Prob(Y_{m_A}=1 \mid X_A=x_A)\\
      &= \frac{f(x_A \mid Y_{m_A}=1)\Prob(Y_{m_A}=1)}
      {f(x_A \mid Y_{m_A}=1) \Prob(Y_{m_A}=1) +f(x_A \mid Y_{m_A}=0) \Prob(Y_{m_A}=0)}\\
      &= \frac{p_A(x_A \mid M_A=\mathbf{0}) \pi_{m_A}}{p_A(x_A \mid M_A=\mathbf{0})\pi_{m_A} + h_{m_A}(x_A)\pi_{m_A}}, 
\end{align*}
such that if $\pi_{m_A}= 0.5$, we obtain that $  \pi^{Y_{m_A}}(x_A) =   \pi_{m_A}(x_A)$ for any observation $x$. 
\end{proof}

\section{Details on Propriety Assessment}\label{app:propriety}

We used the following pipeline to obtain the results of Section 7.3:
\begin{enumerate}
    \item For the two missingness mechanisms MAR and MCAR we consider two overall probabilities of missingness ($p_{\text{miss}}= 0.1$ and $0.2$). For each of the $p_\text{miss}$ we do:
    \item For each fully observed data set \texttt{X} we do:
    \begin{enumerate}
    \item We mask \texttt{X} respecting the missingness mechanism as well as $p_{miss}$ and obtain a fixed \texttt{X.NA}.
    \item We apply $30$ times $1/2$-subsampling on \texttt{X} and \texttt{X.NA} and obtain $30$ subsampled \texttt{X}$^S$ and \texttt{X.NA}$^S$.
    \item We score \texttt{X} with the DR I-Score.
    \item For each method in methods we do: 
    \begin{enumerate}
        \item We impute \texttt{X.NA} $m=5$ times.
    \item We score each of the $5$ imputed versions of \texttt{X.NA} with the DR I-Score and get the final score by averaging.
    \item We compute the difference $D$ of the DR I-Score of the imputation of \texttt{X.NA} and the DR I-Score of \texttt{X}.
    \item For each of the $S=1,\ldots,30$ we do:
    \begin{enumerate}
    \item We apply steps (c) and i. and ii. to \texttt{X}$^S$ and  \texttt{X.NA}$^S$.
    \item We compute the difference $D^S$ of the DR I-Score of the imputation of \texttt{X.NA}$^S$ and the DR I-Score of \texttt{X}$^S$.
    \end{enumerate}
    \item We estimate the variance of $D$, $\sigma^2(D)$, with the Jackknife variance estimator formula (5.2) based on $D^S$.
    \end{enumerate}
    \item We compute a $p$-value by $P_{H_0}(X > D)$, where $X \overset{H_0}{\sim} \mathcal{N}(0,\sigma^2(D)) $
\end{enumerate}

\end{enumerate}

\section{Empirical Results for $p_{\text{miss}}=0.1$} \label{app:pmisses}

In this section we present additional results for $p_{\text{miss}}=0.1$. 

\begin{figure}[H]
     \centering
     \begin{subfigure}[b]{0.48\textwidth}
         \centering
         \includegraphics[width=0.97\textwidth]{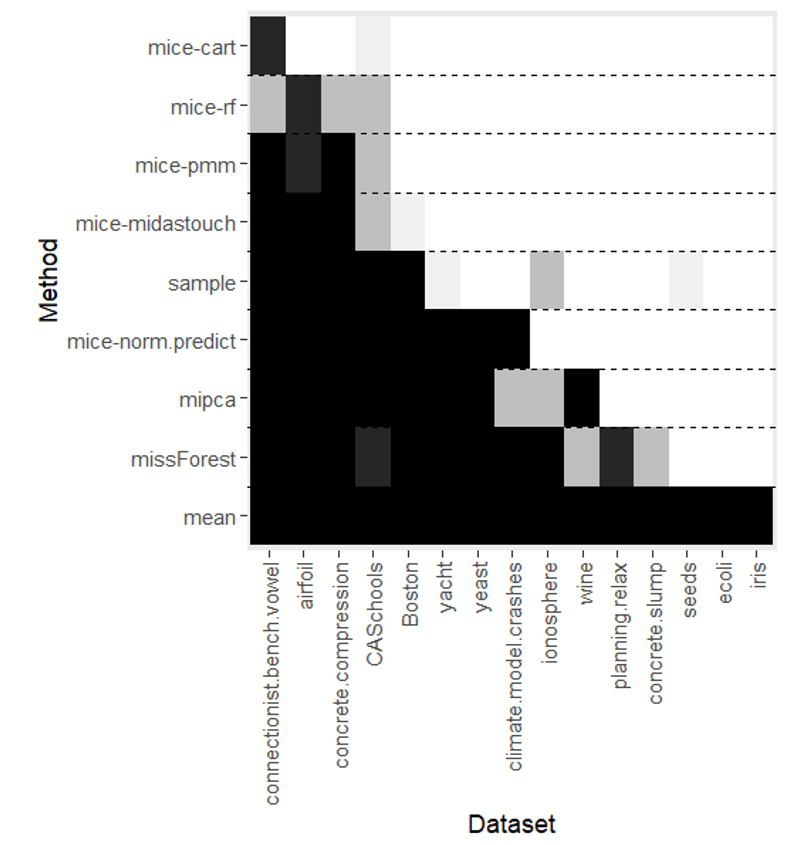}
              \caption{MAR}
     \end{subfigure}
     \begin{subfigure}[b]{0.48\textwidth}
         \centering
           \includegraphics[width=1.2\textwidth]{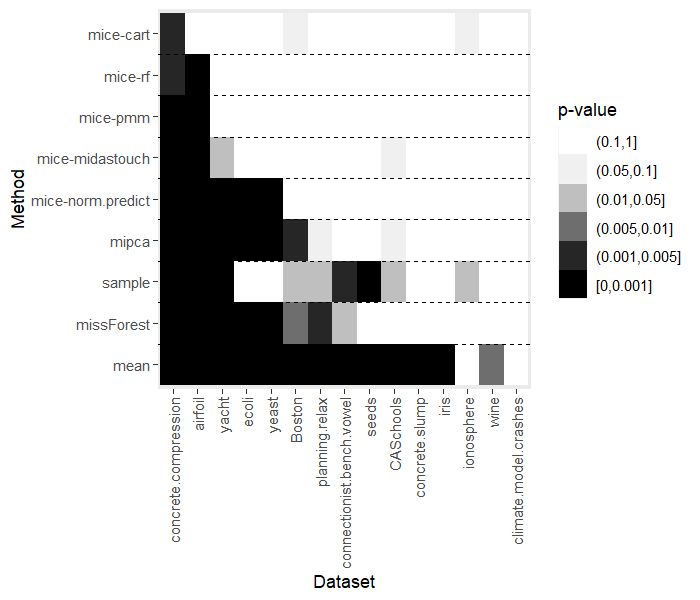}
         \caption{MCAR}
     \end{subfigure}
   \caption{Discretized $p$-values of testing (7.3) under assumption (7.2) for the $9$ methods applied to the $15$ data sets. We used missingness mechanisms MAR (a) and MCAR (b), $p_{miss}=0.1$ and $m=5$. The parameter values of the DR I-Score are described in the main text. }
\end{figure}

\begin{figure}[H]
     \centering
     \begin{subfigure}[b]{0.46\textwidth}
         \centering
         \includegraphics[width=1.1\textwidth]{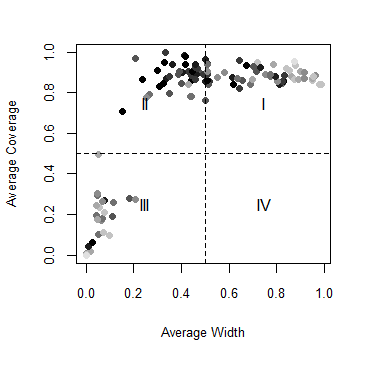}
         \vspace*{-1.2cm}
              \caption{MAR}
     \end{subfigure}
     \begin{subfigure}[b]{0.46\textwidth}
         \centering
           \includegraphics[width=1.1\textwidth]{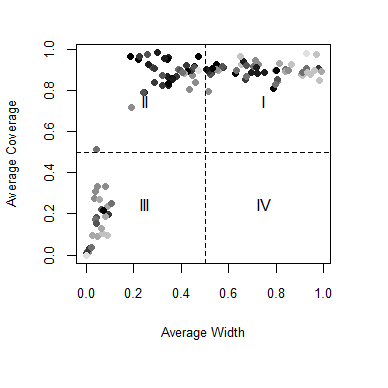}
           \vspace*{-1.2cm}
         \caption{MCAR}
     \end{subfigure}
 \caption{Average coverage plotted against average width for the $9$ methods applied to the $15$ data sets (total = $9\times15=135$ points). The darkness indicates the rank induced by the DR I-Score (the darker, the higher the rank). We used the missingness mechanism MAR in (a) and MCAR in (b) with $p_{miss}=0.1$, $m=20$ and the in the text described parameter values to compute the DR I-Scores.}
 \label{supp:bobeli}
\end{figure}

 \begin{figure}[H]
\begin{minipage}{0.45\textwidth}
\centering
\resizebox{1\textwidth}{!}{
\begin{tabular}{|c|c c  c|}
  \hline
method/quadrant & I    & II    & III \\[0.5ex] 
\hline\hline
cart & 0.20 & 0.80 & 0 \\ 
  pmm & 0.40 & 0.60 & 0 \\ 
  midastouch & 0.40 & 0.60 & 0 \\ 
  rf & 0.67 & 0.33 & 0 \\ 
  mipca & 0.87 & 0.13 & 0 \\ 
  sample & 1.00 & 0 & 0 \\ 
  norm.predict & 0 & 0 & 1 \\ 
  mean & 0 & 0 & 1 \\ 
  missForest & 0 & 0 & 1 \\ 
   \hline
\end{tabular}
}
\subcaption{MAR}
\end{minipage}
\begin{minipage}{0.45\textwidth}
 \resizebox{1\textwidth}{!}{
\begin{tabular}{|c|c c  c|}
  \hline
method/quadrant & I    & II    & III \\[0.5ex] 
\hline\hline
pmm & 0.33 & 0.67 & 0 \\ 
  cart & 0.33 & 0.67 & 0 \\ 
  midastouch & 0.40 & 0.60 & 0 \\ 
  rf & 0.53 & 0.47 & 0 \\ 
  mipca & 0.87 & 0.13 & 0 \\ 
  sample & 0.93 & 0.07 & 0 \\ 
  missForest & 0 & 0.07 & 0.93 \\ 
  norm.predict & 0 & 0 & 1 \\ 
  mean & 0 & 0 & 1 \\ 
   \hline
\end{tabular}
}
\subcaption{MCAR}
\end{minipage}
\captionof{table}{The fraction of times each method appeared in the quadrants I, II and III of Figure \ref{supp:bobeli} in the MAR case (a) and the MCAR case (b).}

\end{figure}

\section*{Supplementary Material}

\textbf{Computer Codes and Data for ``Imputation Scores''.} Computer codes and data used in ``Imputation Scores'' are contained in this zip-folder.

\end{appendix}










\newpage
\bibliographystyle{apalike}
\bibliography{biblio}

\end{document}